\documentclass[runningheads,a4paper]{llncs}

\usepackage{amssymb}
\usepackage{graphicx}
\usepackage{comment}
\usepackage{enumitem}
\usepackage{lipsum}
\usepackage{ntheorem}
\usepackage{caption}

\usepackage[breaklinks]{hyperref}
\newcommand{\orcid}[1]{\href{https://orcid.org/#1}{\includegraphics[width=8pt]{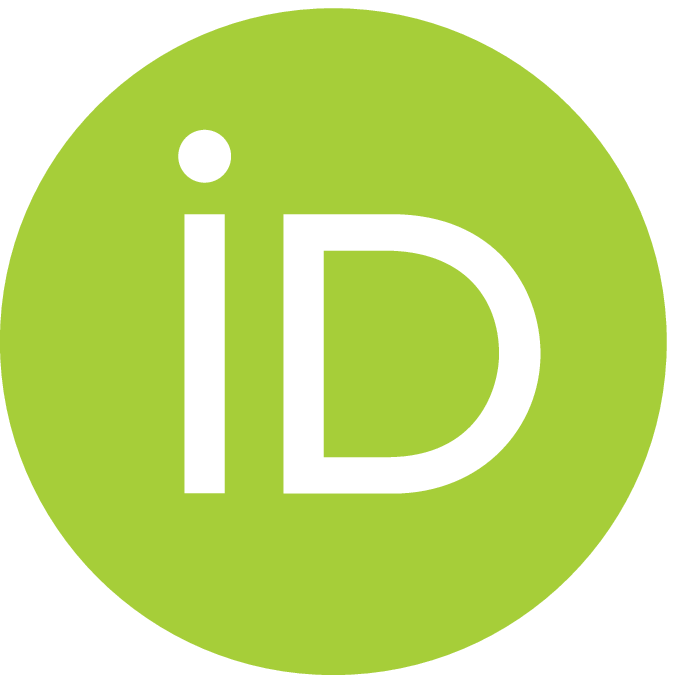}}}

\usepackage{url}
\urldef{\mailsa}\path|shenmaier@mail.ru|

\usepackage{natbib}
\setcitestyle{authoryear,open={(},close={)},aysep={}}

\theorembodyfont{\upshape}
\newtheorem{prob}{Problem}
\newtheorem*{example*}{Example.}
\theorembodyfont{\itshape}
\newtheorem*{definition*}{Definition.}
\newtheorem{fact}{Fact}

\AtBeginDocument{%
  \setlength{\oddsidemargin}{\dimexpr(\paperwidth-\textwidth)/2-1in}%
  \setlength{\evensidemargin}{\oddsidemargin}%
  \setlength{\topmargin}{%
    \dimexpr(\paperheight-\textheight)/2-\headheight-\headsep-1in}%
}

\renewenvironment{abstract}{
  \small
  \list{}{
    \setlength{\leftmargin}{.5cm}%
    \setlength{\rightmargin}{\leftmargin}%
  }%
  \item\relax
}

\begin{document}

\mainmatter  

\title{Linear-Size Universal Discretization of Geometric Center-Based Problems in Fixed Dimensions}

\titlerunning{Linear-Size Universal Discretization of Center-Based Problems}

%
%
\author{Vladimir Shenmaier \orcid{0000-0002-4692-1994}}
%
\authorrunning{V.\,V. Shenmaier}

\institute{Sobolev Institute of Mathematics, Novosibirsk, Russia\\
\mailsa}

%
%

\tocauthor{Vladimir Shenmaier}
\maketitle

\begin{abstract}{\bf Abstract.}
Many geometric optimization problems can be reduced to finding points in space (centers) minimizing an objective function which continuously depends on the distances from the centers to given input points.
Examples are $k$-Means, Geometric $k$-Median/Center, Continuous Facility Location, $m$-Variance, etc.
We prove that, for any fixed $\varepsilon>0$, every set of $n$ input points in fixed-dimensional space with the metric induced by any vector norm admits a set of $O(n)$ candidate centers which can be computed in almost linear time and which contains a $(1+\varepsilon)$-approximation of each point of space with respect to the distances to all the input points.
It gives a universal approximation-preserving reduction of geometric center-based problems with arbitrary continuity-type objective functions to their discrete versions where the centers are selected from a fairly small set of candidates.
The existence of such a linear-size set of candidates is also shown for any metric space of fixed doubling dimension.
\begin{keywords}
Geometric optimization $\cdot$ Clustering $\cdot$ Facility location $\cdot$ Metric space $\cdot$ Approximate centers $\cdot$ Discretization
\end{keywords}
\end{abstract}

\section{Introduction}
We study the following concept, which may be useful for developing approximation algorithms with performance guarantees for geometric center-based problems.
Given a finite set $X$ in a metric space $({\cal M},dist)$, a \emph{$(1+\varepsilon)$-approximate centers collection} or, shortly, a \emph{$(1+\varepsilon)$-collection} for the set $X$ is a subset of ${\cal M}$ which, for every point $p\in{\cal M}$, contains a point $p'$ such that the distance from $p'$ to each element of $X$ is at most $1+\varepsilon$ of that from $p$.

A $(1+\varepsilon)$-collection contains approximations of all the points of ${\cal M}$ with respect to the distances to all the given points.
In particular, it contains approximate solutions to any geometric optimization problem reducible to finding points in space (centers) with the minimum value of arbitrary objective function which has a continuity-type dependence on the distances from the centers to given points.
In fact, a finite $(1+\varepsilon)$-collection is a universal discretization of the space $({\cal M},dist)$ which contains approximate minimums of all such functions.

Geometric center-based problems we describe typically arise in clustering, pattern recognition, facility location, etc.
In these problems, \emph{centers} may be any points of space and may have various practical meanings, e.g., represent sources of detected signals or locations for placing facilities.
In general, problems of finding optimal geometric centers can be written in the following form:\medskip

\noindent\textbf{Geometric Center-Based problem.}
Given a finite set $X$ in a metric space $({\cal M},dist)$, a finite set of positive integers $K$, and a non-negative function $f$ which is defined for any finite set of points and any $k$-element tuple of centers, where $k\in K$.
Find a tuple of centers $c_1,\dots,c_k\in{\cal M}$, $k\in K$, to minimize the value of $f(X;c_1,\dots,c_k)$.\medskip

The mentioned continuity-type dependence of the objective function on the point-to-center distances means that a small relative increase in the distances from the centers to input points must give a bounded relative increase in the objective function value.
This can be formalized as follows:

\begin{definition*}
A non-negative function $f$ satisfies the \emph{continuity-type condition} if there exists a mapping {\rm(}modulus\/{\rm)} $\mu:[1,\infty)\rightarrow [1,\infty)$ such that, for any finite set \mbox{$X\subseteq{\cal M}$}, centers $c_1,c'_1,\dots,c_k,c'_k\in{\cal M}$, and real value $\varepsilon>0$ which fulfill all the inequalities $dist(x,c'_i)\le(1+\varepsilon)\,dist(x,c_i)$, $x\in X$, $i=1,\dots,k$, we have
$$f(X;c'_1,\dots,c'_k)\le\mu(1+\varepsilon)\,f(X;c_1,\dots,c_k).$$
\end{definition*}

It seems to be a natural condition for geometric center-based problems.
As an example, the objective functions of the $k$-Means \citep{Matousek,KumarA2010,KumarA2014}, Geometric $k$-Me\-di\-an \citep{Arora,BHI,Chen2006}, Geometric $k$-Cen\-ter \citep{AP,BHI,KumarP2003}, Continuous Facility Location \citep{Meira}, $m$-Va\-ri\-ance \citep{Aggarwal,EE,Shen2012}, Smallest $m$-En\-c\-lo\-sing Ball
\citep{AHV2005,Shen2015}
problems satisfy this condition with modulus $\mu(1+\varepsilon)=(1+\varepsilon)^g$, where $g\in\{1,2\}$.

A finite $(1+\varepsilon)$-collection provides an approximation-preserving reduction of Geometric Center-Based problem with any continuity-type objective function to the discrete version of this problem where the centers are selected from a finite set:\medskip

\noindent\textbf{Discrete Center-Based problem.}
Given a finite set $X$ in a metric space $({\cal M},dist)$, a finite set of positive integers $K$, and a finite set of feasible centers $Y\subseteq{\cal M}$.
Find a tuple of centers $c_1,\dots,c_k\in Y$, $k\in K$, to minimize the value of $f(X;c_1,\dots,c_k)$.\medskip

Indeed, the definition of a $(1+\varepsilon)$-collection combined with the continuity-type condition immediately implies that an optimum solution $c_1,\dots,c_k$ to the instance of Discrete Center-Based problem with the set $Y$ formed by a $(1+\varepsilon)$-collection for the set $X$ is a $\mu(1+\varepsilon)$-approximate solution to Geometric Center-Based problem.
Si\-mi\-lar\-ly, for any $\beta\ge 1$, a $\beta$-approximate solution to that instance of the discrete problem is also a $\beta\mu(1+\varepsilon)$-approximate solution to the geometric problem.
In particular, if the number of desired centers is bounded by a small constant, then a brute force enumeration of small subsets of a $(1+\varepsilon)$-collection for $X$ gives a $\mu(1+\varepsilon)$-approximate solution to Geometric Center-Based problem.

Note that such a universal instrument as $(1+\varepsilon)$-collections may be relevant in cases when the known fast methods of generating candidate centers are not applicable or do not give desired approximation guarantees:\medskip
\begin{itemize}[label={$\bullet$}]
\item 
The objective function has a more complex form than the classical sums or ma\-xi\-ma of the distances from the centers to input points or the sums of the squared distances.
In general, it may be arbitrary continuity-type function of the point-to-center distances.\smallskip
\item It is required to choose a given-size subset of input points to minimize some continuity-type objective function on this subset.
Such a requirement is typical for problems with ``outliers'' or ``partial covering''.
In this case, standard techniques based on random sampling are not effective since the given size of a desired subset may be arbitrarily small, so any constant number of random samples may ``miss'' good subsets.\smallskip
\item More complex constraints are imposed on the service of the input points by the desired centers than those in usual models, e.g., each input point may be served by a given number of centers, each center has its own capacity and service radius, each point-to-center connection has its own unit distance cost, etc.\smallskip
\item The problem has one or multiple objectives, possibly not specified explicitly, and an oracle is given which, for any two tuples of centers, answers which tuple is better.
In this case, if all the objectives satisfy the continuity-type condition for some modulus $\mu$, then enumerating all the tuples of points in a $(1+\varepsilon)$-collection for the input set gives a $\mu(1+\varepsilon)$-approximate solution to the problem.
\end{itemize}\medskip

\noindent\textbf{Related work.}
The concept of an $\alpha$-collection was introduced in \citep{Shen2019,Shen2020} for the case of Euclidean metric.
In these works, it was suggested an algorithm which computes polynomial-size $(1+\varepsilon)$-collections in high dimensions:
For any fixed $\varepsilon\in(0,1]$ and any set of $n$ points in space $\mathbb R^d$, this algorithm outputs a Euclidean $(1+\varepsilon)$-collection of cardinality
$N(n,\varepsilon)=O\big(n^{\lceil\log_2(2/\varepsilon)/\varepsilon\rceil}\big)$ in time $O(N(n,\varepsilon)\,d\,)$.
It is interesting that the obtained cardinality $N(n,\varepsilon)$ does not depend on the dimension of space.
On the other hand, it was shown that, for any fixed $\varepsilon>0$, the minimum cardinality of a $(1+\varepsilon)$-collection for a given set of $n$ points in high-dimensional Euclidean space is $\Omega\big(n^{\lfloor 1/(16\varepsilon)+1\rfloor}\big)$ in the worst case \citep{Shen2020}.

As an application of $(1+\varepsilon)$-collections, approximation algorithms were obtained for the following clustering problems which contain $k$-Means with outliers, Geometric $k$-Median with outliers, and their versions with cardinality constraints:

\begin{prob}\label{pr1}
Given points $x_1,\dots,x_n$ in space $\mathbb R^d$, integers $k,m\ge 1$, unit distance costs $f_{ij}\ge 0$, and degrees $g_{ij}\in[0,g]$, $i=1,\dots,k$, $j=1,\dots,n$, where $g$ is some parameter.
Find disjoint subsets $S_1,\dots,S_k\subseteq\{1,\dots,n\}$ of total cardinality $m$ and select centers $c_1,\dots,c_k\in\mathbb R^d$ to minimize the value of
$$\sum_{i=1}^k\sum_{j\in S_i}f_{ij}\,dist(x_j,c_i)^{g_{ij}}.$$
\end{prob}

\begin{prob}\label{pr2}
The same as in Problem~\ref{pr1} except that each subset $S_i$ is required to have its own given cardinality $m_i$, $i=1,\dots,k$.
\end{prob}

\begin{fact}{\rm\citep{Shen2019}}\label{f1}
If the values of $k$ and $g$ are fixed, Problems {\rm\ref{pr1}} and~{\rm\ref{pr2}} admit deterministic approximation schemes PTAS computing $(1+\varepsilon)^g$-approximate solutions to these problems in time $O(N(n,\varepsilon)^knkd\,)$ and $\;O(N(n,\varepsilon)^k(n^3+nkd\,))$, respectively.
\end{fact}

Fact~\ref{f1} was proved for the case of high-dimensional Euclidean space but it can be easily extended to any metric space which admits a polynomial upper bound $N(n,\varepsilon)$ for the size of $(1+\varepsilon)$-collections.

It should be noted that not every metric induced by a vector norm admits polynomial $(1+\varepsilon)$-collections in high dimensions.
In particular, in the case of Chebyshev's norm $\ell_\infty$, the size of $(1+\varepsilon)$-collections for some $n$-element sets in space $\mathbb R^{n/2}$ can not be less than $2^{n/2}$ if $\varepsilon<1$ \citep{Shen2020}.
In this regard, an interesting question is the cardinality of $(1+\varepsilon)$-collections in fixed and logarithmic dimensions.\medskip

\noindent\textbf{Our contributions.}
Obviously, the efficiency of a $(1+\varepsilon)$-approximate centers collection depends on its size and the time required to calculate it.
In this paper, we study the question of the existence and computability of small $(1+\varepsilon)$-collections in low-dimensional metric spaces.

We prove that, for any $\varepsilon\in(0,1]$, every set of $n$ points in space $\mathbb R^d$ with the metric induced by arbitrary vector norm admits a $(1+\varepsilon)$-collection of cardinality $N(n,\varepsilon)=(2/\varepsilon)^{O(d)}n$, which can be constructed by a randomized algorithm in expected time $2^{O(d)}n\ln n+(2/\varepsilon)^{O(d)}n$.
For the special case when the metric is Euclidean, a deterministic algorithm is proposed which constructs a $(1+\varepsilon)$-collection of cardinality
$N(n,\varepsilon)=2^{O(d)}(1/\varepsilon)^{2d}\ln(2/\varepsilon)\,n$
in time $O(n\ln n)+2^{O(d)}(1/\varepsilon)^{2d}\ln(2/\varepsilon)\,n$.
The suggested algorithms are based on geometric properties of $(1+\varepsilon)$-collections as well as on known facts on coverings of convex bodies \citep{VempalaExt,VempalaSTOC} and well-separated pair decompositions in the general and Euclidean cases \citep{Mendel,CK}.

The obtained results are partially extended to the case of any metric space $({\cal M},dist)$ of doubling dimension $dim$.
We prove that every set of $n$ points in this space admits a $(1+\varepsilon)$-collection of cardinality $N(n,\varepsilon)=(2/\varepsilon)^{O(dim)}n$.
If an oracle can be specified which, in time $2^{O(dim)}$, returns a covering of arbitrary ball in the space $({\cal M},dist)$ by $2^{O(dim)}$ balls of half the radius, then this $(1+\varepsilon)$-collection can be computed by a ran\-do\-mized algorithm in expected time $2^{O(dim)}n\ln n+(2/\varepsilon)^{O(dim)}n$.

Thus, we state that, for any fixed $\varepsilon\in(0,1]$, every set of $n$ points in a metric space of fixed doubling dimension admits a linear-size $(1+\varepsilon)$-approximate centers collection, moreover, at least in the case when the metric is induced by a vector norm in fixed-dimensional space $\mathbb R^d$, such a $(1+\varepsilon)$-collection can be computed in almost linear time.
Note that the proposed algorithms remain to be polynomial not only in fixed but also in logarithmic dimensions, i.e., when $d=O(\ln n)$.
In this case, a $(1+\varepsilon)$-collection of size $n^{O(\ln(2/\varepsilon))}$ can be computed in time $n^{O(\ln(2/\varepsilon))}$.

To be objective, the best known techniques for some classical center-based problems in Euclidean space, such as $k$-Means, Geometric $k$-Median, Continuous Facility Location, give smaller sets of approximate centers \citep[e.g., see][]{Matousek,Mazumdar,KumarA2010,Meira}.
But an advantage of $(1+\varepsilon)$-approximate centers collections is that they do not depend on any specific problem:
they approximate optimal centers not by the objective function values but by the distances from the candidate centers to every input point.
So they contain approximate centers at once for all the objective functions which have a continuity-type dependence on these distances.
As a result, it immediately gives approximation algorithms for a wide range of geometric center-based problems, including those which cannot be approximated by the known techniques.
Examples are Problems \ref{pr1} and~\ref{pr2}, for which the proposed framework gives approximation schemes \mbox{FPTAS} in the case of fixed parameters $k,g,d$.
Another advantage of our approach is that the algorithms we suggest construct small $(1+\varepsilon)$-collections for any metric, not only Euclidean, which further expands the scope of this instrument.

\section{Basic definitions and properties}
First, we give basic definitions and properties related to $\alpha$-approximate centers collections which underly the suggested algorithms.
Let $X$ be arbitrary set of $n$ points in any metric space $({\cal M},dist)$.
It is assumed that the distance function $dist$ satisfies the triangle inequality and the symmetry axiom.

\begin{definition*}
Given a real number $\alpha\ge 1$ and a point $p\in{\cal M}$, an \emph{$\alpha$-approximation of $p$} with respect to the set $X$ in the metric $dist$ is any point $p'\in{\cal M}$ such that $dist(x,p')\le\alpha\,dist(x,p)$ for all $x\in X$.
\end{definition*}

\begin{definition*}
Given a real number $\alpha\ge 1$, an \emph{$\alpha$-approximate centers collection} or, shortly, an \emph{$\alpha$-collection} for the set $X$ in the metric space $({\cal M},dist)$ is a subset of ${\cal M}$ which contains $\alpha$-approximations of all the points of ${\cal M}$ with respect to the set $X$ in the metric $dist$.
\end{definition*}

In what follows, we will omit the words ``in the metric $dist$'' and ``in the metric space $({\cal M},dist)$'' since both metric and metric space will always be clear by context.

\begin{example*}
Any finite set $X\subseteq{\cal M}$ is a $2$-collection for itself.
Indeed, let $p\in{\cal M}$ and $p'$ be a point of $X$ nearest to~$p$.
Then, by the triangle inequality, the choice of $p'$, and the symmetry axiom, we have
$$dist(x,p')\le dist(x,p)+dist(p,p')\le dist(x,p)+dist(p,x)=2\,dist(x,p)$$
for all $x\in X$.
So $p'$ is a $2$-approximation of $p$ with respect to~$X$.
\end{example*}

Our further goal is constructing $\alpha$-collections in the case of smaller values of $\alpha$, when $\alpha=1+\varepsilon$ for $\varepsilon\in(0,1)$.

Given a point $x\in{\cal M}$ and a real number $r>0$, denote by $B(x,r)$ the ball of radius $r$ in the space $({\cal M},dist)$ centered at $x$:
$B(x,r)=\{p\in{\cal M}\,|\,dist(x,p)\le r\}$.
Next, given $x_1,x_2\in X$ and $\varepsilon>0$, define the symmetric lens
$$L_\varepsilon(x_1,x_2)=B(x_1,r)\cap B(x_2,r),\mbox{ where }r=dist(x_1,x_2)/(1+\varepsilon).$$

\begin{lemma}\label{lem1}
Suppose that $p$ is any point of ${\cal M}$, $x$ is an element of $X$ nearest to $p$, and $p$ does not belong to any lens $L_\varepsilon(x,y)$, $y\in X$.
Then $x$ is a $(1+\varepsilon)$-approximation of $p$ with respect to~$X$.
\end{lemma}

\begin{proof}
Let $y$ be any element of $X$ and $r=dist(x,y)/(1+\varepsilon)$.
Then, by the condition, the point $p$ does not belong to at least one of the balls $B(x,r)$ and $B(y,r)$.
But $dist(x,p)\le dist(y,p)$, so $p\notin B(y,r)$.
It follows that $dist(y,p)>dist(x,y)/(1+\varepsilon)$ or, equivalently, $dist(y,x)<(1+\varepsilon)\,dist(y,p)$.
Thus, $x$ is a $(1+\varepsilon)$-approximation of $p$ with respect to~$X$.
The lemma is proved.
\hfill$\Box$
\end{proof}

\begin{lemma}\label{lem2}
Suppose that $p,p'\in{\cal M}$ and $dist(p,p')\le\varepsilon\,dist(x,p)$, where $x$ is a point of $X$ nearest to $p$.
Then $p'$ is a $(1+\varepsilon)$-approximation of $p$ with respect to~$X$.
\end{lemma}

\begin{proof}
Let $y$ be any point of $X$.
Then, by the triangle inequality and the choice of $x$ and $p'$, we have
$$dist(y,p')\le dist(y,p)+dist(p,p')\le dist(y,p)+\varepsilon\,dist(x,p)\le(1+\varepsilon)\,dist(y,p).$$
The lemma is proved.
\hfill$\Box$
\end{proof}

Given a set $S\subseteq X$ and a point $p\in{\cal M}$, define the value $\displaystyle dist(S,p)=\min_{s\in S}dist(s,p)$.

\begin{lemma}\label{lem3}
Suppose that, for some finite set $C\subseteq{\cal M}$, every element $p$ of each lens $L_\varepsilon(x_1,x_2)$, $x_1,x_2\in X$, lies at distance at most $\varepsilon\,dist(\{x_1,x_2\},p)$ from $C$.
Then $X\cup C$ is a $(1+\varepsilon)$-collection for~$X$.
\end{lemma}

\begin{proof}
Let $u$ be any point of ${\cal M}$ and $x$ be an element of $X$ nearest to~$u$.
If $u$ does not belong to any lens $L_\varepsilon(x,y)$, $y\in X$, then $x$ is a $(1+\varepsilon)$-approximation of $u$ by Lemma~\ref{lem1}.
Suppose that $u$ belongs to at least one of such lenses and $u'$ is a point of $C$ nearest to~$u$.
Then, by the condition and the choice of $x$, we have $dist(u,u')\le\varepsilon\,dist(x,u)$.
By Lemma~\ref{lem2}, it follows that $u'$ is a $(1+\varepsilon)$-approximation of $u$ with respect to~$X$.
Thus, the set $X\cup C$ is a $(1+\varepsilon)$-collection for $X$.
The lemma is proved.
\hfill$\Box$
\end{proof}

\section{Quadratic-size $(1+\varepsilon)$-collections}
Lemma~\ref{lem3} prompts an idea how to construct a $(1+\varepsilon)$-collection for the given set $X$:
If we cover each lens $L_\varepsilon(x_1,x_2)$, $x_1,x_2\in X$, by sufficiently small balls with radii proportional to the distances from their centers to the points $x_1$ and $x_2$, then the centers of these balls, supplemented by the elements of $X$, form a $(1+\varepsilon)$-collection for~$X$.
Such a $(1+\varepsilon)$-collection will be of size $O(n^2)$ if the used covering of each lens is of size $O(1)$ for fixed $\varepsilon$.
Later (see Sect.~4), we will suggest a method to reduce this construction to a linear-size one.

\subsection{Covering the lenses $L_\varepsilon(x_1,x_2)$}
Here, we describe a very simple scheme of covering the lenses $L_\varepsilon(x_1,x_2)$ which allows to satisfy the condition of Lemma~\ref{lem3}.
First, for each $x_1,x_2\in X$, we define the values
$$\textstyle d_i(x_1,x_2)=dist(x_1,x_2)\frac{\varepsilon^{1-i/I}}{1+\varepsilon},$$
where $I$ is some positive integer parameter and $i=0,\dots,I$.
These values approximate the possible distances from $x_1$ and $x_2$ to the points of $L_\varepsilon(x_1,x_2)$.
Then, for each $x_j$, $j=1,2$, we consider the balls $B(x_j,d_i(x_1,x_2))$, $i=1,\dots,I$ (see Fig.~\ref{fig:1}), and cover each of them by balls of the proportional radius $\delta d_i(x_1,x_2)$, where $\delta=\varepsilon^{1+1/I}$.

\begin{figure}[!ht]
\centering
\captionsetup{justification=centering}
\includegraphics[scale=0.9]{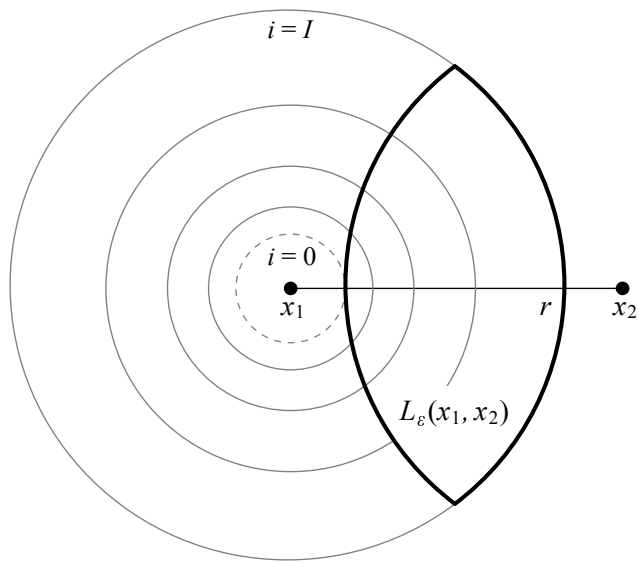}
\caption{The lens $L_\varepsilon(x_1,x_2)$ and the balls $B(x_1,d_i(x_1,x_2))$}
\label{fig:1}
\end{figure}

\begin{lemma}\label{lem4}
Suppose that a finite set $C\subseteq{\cal M}$ satisfies the property
$$B(x_1,\,d_i(x_1,x_2))\subseteq\bigcup_{c\in C}B(c,\,\delta d_i(x_1,x_2))$$
for each $x_1,x_2\in X$ and\/ $i=1,\dots,I$.
Then $X\cup C$ is a $(1+\varepsilon)$-collection for~$X$.
\end{lemma}

\begin{proof}
To check that the condition of Lemma~\ref{lem3} is satisfied, we consider arbitrary point $p$ in any lens $L_\varepsilon(x_1,x_2)$, $x_1,x_2\in X$.
We will assume that $dist(x_1,p)\le dist(x_2,p)$:
the opposite case is treated similarly, using $x_2$ instead of $x_1$ and $x_1$ instead of $x_2$.

By the construction of the lens $L_\varepsilon(x_1,x_2)$, we have $dist(x_j,p)\le r$, where $j=1,2$ and $r=dist(x_1,x_2)/(1+\varepsilon)$.
Then
$$dist(x_1,x_2)\le dist(x_1,p)+dist(x_2,p)\le dist(x_1,p)+r.$$
Taking into account the definition of $d_i(x_1,x_2)$, it follows that
$$\textstyle d_0(x_1,x_2)=dist(x_1,x_2)-r\le dist(x_1,p)\le r=d_I(x_1,x_2),$$
which implies the existence of an index $i\in\{1,\dots,I\}$ with the property
$$d_{i-1}(x_1,x_2)\le dist(x_1,p)\le d_i(x_1,x_2).$$
So $d_i(x_1,x_2)=d_{i-1}(x_1,x_2)\,\varepsilon^{-1/I}\le dist(x_1,p)\,\varepsilon^{-1/I}$ and then, by the condition, the distance from $p$ to $C$ is at most
$$\delta d_i(x_1,x_2)\le\delta\varepsilon^{-1/I}dist(x_1,p)=\varepsilon\,dist(x_1,p)=\varepsilon\,dist(\{x_1,x_2\},p).$$
Therefore, by Lemma~\ref{lem3}, the set $X\cup C$ is a $(1+\varepsilon)$-collection for $X$.
The lemma is proved.
\hfill$\Box$
\end{proof}

Thus, we reduce computing a $(1+\varepsilon)$-collection for the set $X$ to covering the balls $B(x_1,d_i(x_1,x_2))$, $x_1,x_2\in X$, $i=1,\dots,I$, by balls of radius $\delta d_i(x_1,x_2)$.

\subsection{Covering the ball $B({\textbf 0},1)$ in a normed space}
Based on Lemma~\ref{lem4}, we will construct a $(1+\varepsilon)$-collection for the given set $X$ in the case when ${\cal M}=\mathbb R^d$ and the given metric $dist$ is induced by arbitrary vector norm $\|.\|$ in space $\mathbb R^d$: $dist(x,y)=\|x-y\|$ for all $x,y\in\mathbb R^d$.
In this case, the metric $dist$ is translation invariant and homogeneous, so covering a ball of radius $d_i(x_1,x_2)$ by balls of radius $\delta d_i(x_1,x_2)$ is reduced to covering the ball $B({\textbf 0},1)$ by translates of the ball $B({\textbf 0},\delta)$, where $\textbf 0$ is the zero vector.

To cover the unit ball $B({\textbf 0},1)$, we will use known coverings of centrally-symmetric convex bodies by its $\sigma$-scaled copies, $\sigma\in(0,1)$.
Obviously, in the case of Chebyshev's norm $\ell_\infty$, the ball $B({\textbf 0},1)$ can be covered by $\lceil 1/\sigma\rceil^d$ its copies.
In the Euclidean case, by using a volume argument, it can be easily shown that the axis-parallel grid
$$\textstyle Gr(d,\sigma)=\big(\frac{2\sigma}{\sqrt{d}}\big)\,\mathbb Z^d\cap B({\bf 0},1+\sigma)$$
gives a covering of size $2^{O(d)}(1/\sigma)^d$ \citep[e.g., see Lemma~7 in][]{KMS}.
Note that, in both these cases, the coverings of the ball $B({\textbf 0},1)$ can be constructed in time $2^{O(d)}(1/\sigma)^d$ and remain to be polynomial even if the space dimension is not fixed but bounded by a ``slowly growing'' value $\Theta(\ln n)$.

An interesting question is whether such a covering exists in the case of any vector norm in space $\mathbb R^d$.
The following facts imply that the answer to this question is ``yes''.

\begin{fact}{\rm\citep{VempalaExt}}\label{f2}
Let $M(A,B)$ be the minimum number of translates of a body $B$ required to cover a body $A$.
Then, for any convex body $A$ and any centrally symmetric convex body $B$ in space $\mathbb R^d$, a covering of $A$ by $M(A,B)\,2^{O(d)}$ translates of $B$ can be constructed by a deterministic algorithm in time $M(A,B)\,2^{O(d)}$.
\end{fact}

\begin{fact}{\rm\citep[Folklore, see][]{VempalaSTOC}}\label{f3}
For any $\sigma\in(0,1)$ and any convex body $A$ in space $\mathbb R^d$, we have $M(A,\sigma A)\le(1+2/\sigma)^d$, where $\sigma A$ is a $\sigma$-scaled copy of $A$.
\end{fact}

Thus, for any vector norm, a covering of the ball $B({\bf 0},1)$ by $2^{O(d)}(1/\sigma)^d$ translates of the ball $B({\bf 0},\sigma)$ can be constructed by a deterministic $2^{O(d)}(1/\sigma)^d$-time algorithm.

\subsection{Constructing a quadratic-size $(1+\varepsilon)$-collection in a normed space}
Let us summarize the above observations in the form of an algorithm for constructing $(1+\varepsilon)$-collections in space $\mathbb R^d$.
If $\varepsilon\ge 1$, then the set $X$ is a $(1+\varepsilon)$-collection for itself as follows from Example in Sect.~2.
Further, we assume that $\varepsilon<1$.

First, we run the deterministic algorithm from Fact~\ref{f2} which constructs a set $C_\delta$ of cardinality $2^{O(d)}(1/\delta)^d$ such that the balls $B(c,\delta)$, $c\in C_\delta$, cover the ball $B({\bf 0},1)$.
Then, for each pair $x_1,x_2\in X$ and each index $i=1,\dots,I$, we calculate the value $d_i(x_1,x_2)=dist(x_1,x_2)\frac{\varepsilon^{1-i/I}}{1+\varepsilon}$ and add to the output the elements of the set
$$C_i(x_1,x_2)=x_1+d_i(x_1,x_2)\,C_\delta,$$
where $x+rA=\{x+ra\,|\,a\in A\}$ for $x\in\mathbb R^d$, $A\subseteq\mathbb R^d$, and $r\in\mathbb R$.

Note that, since the ball $B({\bf 0},1)$ is covered by the balls $B(c,\delta)$, $c\in C_\delta$, then the ball $B(x_1,d_i(x_1,x_2))=x_1+d_i(x_1,x_2)B({\bf 0},1)$ is covered by the balls $B(c,\delta d_i(x_1,x_2))$, $c\in C_i(x_1,x_2)$.
By Lemma~\ref{lem4}, it follows that the union of the sets $C_i(x_1,x_2)$, $x_1,x_2\in X$, $i=1,\dots,I$, supplemented by the points of $X$, forms a $(1+\varepsilon)$-collection for~$X$.

It remains to select a good value of the parameter $I$.
We set $I=\big\lceil 1/\log_\varepsilon 0.9\big\rceil$.
Then $I<1+\ln\varepsilon/\ln 0.9<1+10\ln(1/\varepsilon)$ and $\delta=\varepsilon^{1+1/I}\ge 0.9\,\varepsilon$.
In this case, the set $C_\delta$ consists of $2^{O(d)}(1/\delta)^d=2^{O(d)}(1/\varepsilon)^d$ elements and can be constructed in time $2^{O(d)}(1/\varepsilon)^d$.
Thus, we obtain a $(1+\varepsilon)$-collection of size
$$n^2\cdot(1+10\ln(1/\varepsilon))\cdot 2^{O(d)}(1/\varepsilon)^d=2^{O(d)}(1/\varepsilon)^d\ln(2/\varepsilon)\,n^2$$
in time $2^{O(d)}(1/\varepsilon)^d\ln(2/\varepsilon)\,n^2$.

\section{Linear-size $(1+\varepsilon)$-collections}
In this section, we describe how to construct a $(1+\varepsilon)$-collection of linear size.
The main idea is that we will process not every lens $L_\varepsilon(x_1,x_2)$, $x_1,x_2\in X$, but only $O(n)$ of them which approximate all the other ones.
A key instrument to get such an approximation is well-separated pair decompositions (WSPD).

\subsection{Basic facts on WSPD}
Let $({\cal M},dist)$ be any metric space and $X$ be a set of $n$ points in this space.
For any finite set $A\subseteq{\cal M}$, denote by $diam(A)$ its \emph{diameter}, i.e., the maximum value of $dist(x,y)$ over all $x,y\in A$.
For any finite sets $A,B\subseteq{\cal M}$, denote by $dist(A,B)$ the minimum value of $dist(x,y)$ over all $x\in A$, $y\in B$ and denote by $A\otimes B$ the set of all the unordered pairs $\{a,b\}$, $a\in A$, $b\in B$, $a\ne b$.

\begin{definition*}
Given a real number $t\ge 1$, a \emph{$t$-well separated pair decompostion} or, shortly, a \emph{$t$-WSPD} of the set $X$ is a family of pairs $\{A_1,B_1\},\dots,\{A_s,B_s\}$ such that\smallskip\\
{\rm({\tt a})} $A_k,B_k\subseteq X$ for every $k\in\{1,\dots,s\}$;\smallskip\\
{\rm({\tt b})} $A_k\cap B_k=\emptyset$ for every $k\in\{1,\dots,s\}$;\smallskip\\
{\rm({\tt c})} $\bigcup_{k=1}^sA_k\otimes B_k=X\otimes X$;\smallskip\\
{\rm({\tt d})} $dist(A_k,B_k)\ge t\max\big\{diam(A_k),\,diam(B_k)\big\}$ for every $k\in\{1,\dots,s\}$.
\end{definition*}

Here, we use the definition of WSPD in the form of \citet{Mendel}.
The stronger form of \citet{CK,Talwar} additionally requires that different pairs $A_k\otimes B_k$ do not intersect.

Intuitively, a $t$-WSPD gives an approximation of all $n(n-1)/2$ pairs $x_1,x_2\in X$ by $s$ ones.
Indeed, by properties ({\tt c}) and ({\tt d}) of WSPD, if $t$ is sufficiently large and a representative $\{a_k,b_k\}$ of each set $A_k\otimes B_k$ is chosen, then every pair $\{x_1,x_2\}\in X\otimes X$ is ``metrically close'' to one of the pairs $\{a_k,b_k\}$ in the sense that the distances between the corresponding elements of these pairs are relatively close to zero.

An important fact is that any finite set in a metric space of fixed doubling dimension admits a linear-size WSPD.

\begin{definition*}
The \emph{doubling dimension} of a metric space is the smallest value $dim\ge 0$ such that every ball in this space can be covered by $2^{dim}$ balls of half the radius.
\end{definition*}

\begin{fact}{\rm\citep{Mendel}}\label{f4}
For any $t\ge 1$ and any set of $n$ points in a metric space of doubling dimension $dim$, a $t$-WSPD of size $t^{O(dim)}n$ can be computed by a randomized algorithm in expected time $2^{O(dim)}n\ln n+t^{O(dim)}n$.
\end{fact}

Note that space $\mathbb R^d$ with the metric induced by any vector norm is of doubling dimension $O(d\,)$.
Indeed, by Fact~\ref{f3} applying to $\delta=1/2$, each ball in this space can be covered by $5^d$ balls of half the radius.
So we have $dim\le\log_2(5^d)\approx 2.32\,d$.
By Fact~\ref{f4}, it follows that, in the case of $d$-dimensional normed vector space, we can construct a $t$-WSPD of size $t^{O(d)}n$ in expected time $2^{O(d)}n\ln n+t^{O(d)}n$.
In the Euclidean case, such a $t$-WSPD can also be computed by a deterministic algorithm:

\begin{fact}{\rm\setcitestyle{open={},close={}}(\citep{CK}; \citep[see also][]{Smid})\setcitestyle{open={(},close={)}}}\label{f5}
For any $t\ge 1$ and any set of $n$ points in space $\mathbb R^d$ with Euclidean metric, a $t$-WSPD of size $t^dn$ can be computed by a deterministic algorithm in time $O(n\ln n+t^dn)$.
\end{fact}

\subsection{From $O(n^2)$ to $O(n)$}
Suppose that we are given a $t$-WSPD $\{A_1,B_1\},\dots,\{A_s,B_s\}$ of the set $X$ for some large value of $t$ which will be specified later.
Next, for each $k=1,\dots,s$, we choose arbitrary representatives $a_k\in A_k$, $b_k\in B_k$.\medskip

\noindent\textbf{Idea of a linear-size set construction.}
As before, we will use Lemma~\ref{lem4}.
To get the set of centers $C$ satisfying its condition, we will construct a slightly excess covering of each ball $B(a_k,d_i(a_k,b_k))$, $k=1,\dots,s$, $i=1,\dots,I$, by balls of a radius slightly less than $\delta d_i(a_k,b_k)$.
According to property ({\tt d}) of WSPD, for large values of $t$, every pair $x_1\in A_k$, $x_2\in B_k$ is ``metrically close'' to the pair $a_k,b_k$, so the constructed covering of the ball $B(a_k,d_i(a_k,b_k))$ will also be a good covering of the ball $B(x_1,d_i(x_1,x_2))$.
Hence, by property ({\tt c}) of WSPD, the union of all the constructed coverings will be a good covering of all the balls $B(x_1,d_i(x_1,x_2))$, $x_1,x_2\in X$.\medskip

Let us give a more detailed description and justification of such a construction.
The condition of Lemma~\ref{lem4} requires that each ball $B(x_1,d_i(x_1,x_2))$ must be covered by balls of radius $\delta d_i(x_1,x_2)$, where $d_i(x_1,x_2)=dist(x_1,x_2)\frac{\varepsilon^{1-i/I}}{1+\varepsilon}$ and $\delta=\varepsilon^{1+1/I}$.
By the triangle inequality and property ({\tt d}) of WSPD, the distance between any two points $x_1\in A_k$, $x_2\in B_k$ is estimated as
$$dist(x_1,x_2)\le dist(a_k,b_k)+dist(x_1,a_k)+dist(b_k,x_2)\le dist(a_k,b_k)(1+2/t),$$
$$dist(x_1,x_2)\ge dist(a_k,b_k)-dist(x_1,a_k)-dist(b_k,x_2)\ge dist(a_k,b_k)(1-2/t),$$
so
$$d_i(a_k,b_k)(1-2/t)\le d_i(x_1,x_2)\le d_i(a_k,b_k)(1+2/t).$$
Property ({\tt d}) of WSPD also implies that each point $x_1\in A_k$ lies at distance at most $dist(a_k,b_k)/t$ from $a_k$.
It follows that
$B(x_1,d_i(x_1,x_2))\subseteq B(a_k,r_i(k,t))$, where
$$r_i(k,t)=d_i(a_k,b_k)(1+2/t)+dist(a_k,b_k)/t.$$
Similarly, we have $B(x_2,d_i(x_1,x_2))\subseteq B(b_k,r_i(k,t))$ for each $x_2\in B_k$.
Thus, to sa\-tis\-fy the condition of Lemma~\ref{lem4}, it is sufficient to cover each of the balls $B(a_k,r_i(k,t))$ and $B(b_k,r_i(k,t))$ by balls of radius $\delta d_i(a_k,b_k)(1-2/t)$:

\begin{lemma}\label{lem5}
Suppose that $t>2$ and a finite set $C\subseteq{\cal M}$ satisfies the property
$$B(x,\,r_i(k,t))\subseteq\bigcup_{c\in C}B\big(c,\,\delta d_i(a_k,b_k)(1-2/t)\big)$$
for each $k=1,\dots,s$, $i=1,\dots,I$, and $x\in\{a_k,b_k\}$.
Then $X\cup C$ is a $(1+\varepsilon)$-collection for~$X$.
\end{lemma}

\begin{proof}
Let $u,v$ be any different elements of $X$ and $i\in\{1,\dots,I\}$.
Property ({\tt c}) of WSPD implies that the pair $\{u,v\}$ belongs to some set $A_k\otimes B_k$, $k\in\{1,\dots,s\}$.
Then, by the above observations, the inequality $d_i(a_k,b_k)(1-2/t)\le d_i(u,v)$ holds and the ball $B(u,d_i(u,v))$ is contained in one of the balls $B(a_k,r_i(k,t))$ and $B(b_k,r_i(k,t))$.
By the condition, it follows that each element of the ball $B(u,d_i(u,v))$ lies at distance at most $\delta d_i(a_k,b_k)(1-2/t)\le\delta d_i(u,v)$ from the set~$C$.
Therefore, according to Lemma~\ref{lem4}, the set $X\cup C$ is a $(1+\varepsilon)$-collection for~$X$.
The lemma is proved.
\hfill$\Box$
\end{proof}

\subsection{Constructing a linear-size $(1+\varepsilon)$-collection in a normed space}
Here, based on Lemma~\ref{lem5}, we describe how to compute a linear-size $(1+\varepsilon)$-collection in the case when ${\cal M}=\mathbb R^d$ and the metric $dist$ is induced by any vector norm.
In this case, covering a ball of radius $r_i(k,t)$ by balls of radius $\delta d_i(a_k,b_k)(1-2/t)$ is reduced to covering the ball $B({\bf 0},1)$ by balls of radius $\delta'=\delta d_i(a_k,b_k)(1-2/t)/r_i(k,t)$.
To construct such a covering, we use the coverings described in Sect.~3.2.

Let us set $I=\big\lceil 1/\log_\varepsilon 0.9\big\rceil$ and $t=\max\big\{10,\,\frac{1+\varepsilon}{\varepsilon}\big\}$.
Then, as before, we obtain the inequalties $I<1+\ln\varepsilon/\ln 0.9<1+10\ln(1/\varepsilon)$ and $\delta=\varepsilon^{1+1/I}\ge 0.9\,\varepsilon$.
On the other hand, by the choice of $t$ and the definition of $d_i(.\,,.)$, we have
$$\textstyle dist(a_k,b_k)/t\le dist(a_k,b_k)\frac{\varepsilon}{1+\varepsilon}=d_0(a_k,b_k)\le d_i(a_k,b_k),$$
so
$$r_i(k,t)\le d_i(a_k,b_k)(1+2/t)+d_i(a_k,b_k)\le 2.2\,d_i(a_k,b_k).$$
It follows that $\delta'\ge\delta\frac{0.8}{2.2}>0.3\,\varepsilon$.
Denote by $C_{\delta'}$ the set of centers of the $\delta'$-radius balls covering $B({\bf 0},1)$ which is constructed by the algorithm from Fact~\ref{f2}.
In the case of Euclidean distances, we will assume that $C_{\delta'}=Gr(d,\delta')$, where the set $Gr(d,.\,)$ is defined in Sect.~3.2, i.e.,
$C_{\delta'}=\big(\frac{2\delta'}{\sqrt{d}}\big)\,\mathbb Z^d\cap B({\bf 0},1+\delta').$

Note that, since the ball $B({\bf 0},1)$ is covered by the balls $B(c,\delta')$, $c\in C_{\delta'}$, then the ball $B(a_k,r_i(k,t))=a_k+r_i(k,t)B({\bf 0},1)$ is covered by the balls of radius
$$\delta'r_i(k,t)=\delta d_i(a_k,b_k)(1-2/t)$$
centered at the elements of the set
$C'_i(a_k,b_k)=a_k+r_i(k,t)\,C_{\delta'}$.
Similarly, the ball $B(b_k,r_i(k,t))$ is covered by the balls of radius $\delta d_i(a_k,b_k)(1-2/t)$ centered at the elements of the set $C'_i(b_k,a_k)=b_k+r_i(k,t)\,C_{\delta'}$.
By Lemma~\ref{lem5}, this implies that the union $C$ of the sets $C'_i(a_k,b_k)\cup C'_i(b_k,a_k)$, $k=1,\dots,s$, $i=1,\dots,I$, supplemented by the points of $X$, forms a $(1+\varepsilon)$-collection for $X$.

Let us estimate the cardinality of the set $C$ and the time required to construct it.
The algorithm from Fact~\ref{f2} computes the set $C_{\delta'}$ in time $2^{O(d)}(1/\delta')^d=2^{O(d)}(1/\varepsilon)^d$ and the size of this set is $2^{O(d)}(1/\varepsilon)^d$.
The randomized algorithm from Fact~\ref{f4} outputs a $t$-WSPD of the set $X$ in expected time $2^{O(d)}n\ln n+t^{O(d)}n=2^{O(d)}n\ln n+(2/\varepsilon)^{O(d)}n$ and the size of this WSPD is $s=t^{O(d)}n=(2/\varepsilon)^{O(d)}n$.
Hence, in the case of arbitrary vector norm, the set $C$ consists of $2sI\cdot 2^{O(d)}(1/\varepsilon)^d=(2/\varepsilon)^{O(d)}n$ elements and can be constructed in expected time $2^{O(d)}n\ln n+(2/\varepsilon)^{O(d)}n$.

In the case of Euclidean metric, we use the set $C_{\delta'}=Gr(d,\delta')$, which can be constructed in time $2^{O(d)}(1/\delta')^d=2^{O(d)}(1/\varepsilon)^d$ and consists of $2^{O(d)}(1/\varepsilon)^d$ points.
A $t$-WSPD of the set $X$ is computed by the deterministic algorithm from Fact~\ref{f5} in time $O(n\ln n+t^dn)=O(n\ln n)+2^{O(d)}(1/\varepsilon)^dn$ and its size is $s=t^dn=2^{O(d)}(1/\varepsilon)^dn$.
Hence, the set $C$ consists of $2sI\cdot 2^{O(d)}(1/\varepsilon)^d=2^{O(d)}(1/\varepsilon)^{2d}\ln(2/\varepsilon)\,n$ elements and can be constructed in time $O(n\ln n)+2^{O(d)}(1/\varepsilon)^{2d}\ln(2/\varepsilon)\,n$.

So we obtain the following statements:

\begin{theorem}
For any $\varepsilon\in(0,1]$, every set of $n$ points in space $\mathbb R^d$ with the metric induced by any vector norm admits a $(1+\varepsilon)$-collection of cardinality $(2/\varepsilon)^{O(d)}n$, which can be computed by a randomized algorithm in expected time $2^{O(d)}n\ln n+(2/\varepsilon)^{O(d)}n$.
\end{theorem}

\begin{theorem}
For any $\varepsilon\in(0,1]$, every set of $n$ points in space $\mathbb R^d$ with Euclidean metric admits a $(1+\varepsilon)$-collection of cardinality $2^{O(d)}(1/\varepsilon)^{2d}\ln(2/\varepsilon)\,n$, which can be computed by a deterministic algorithm in time $O(n\ln n)+2^{O(d)}(1/\varepsilon)^{2d}\ln(2/\varepsilon)\,n$.
\end{theorem}

Thus, if the dimension of space is fixed, the proposed algorithms compute linear-size $(1+\varepsilon)$-collections in almost linear time.
Note that these algorithms remain to be polynomial even if $d$ is not fixed but bounded by a ``slowly growing'' value $\Theta(\ln n)$.
In this case, a $(1+\varepsilon)$-collection of size $(2/\varepsilon)^{O(\ln n)}n=n^{O(\ln(2/\varepsilon))}$ can be constructed in time $2^{O(\ln n)}n\ln n+(2/\varepsilon)^{O(\ln n)}n=n^{O(\ln(2/\varepsilon))}$.

\subsection{Linear-size $(1+\varepsilon)$-collections in a doubling space}
The constructions described in Sect. 4.2 and~4.3 can be extended to the case of any metric space $({\cal M},dist)$ of doubling dimension $dim$.
Such an extension is based on Lemma~\ref{lem5} and the following simple observation:

\begin{lemma}\label{lem6}
Suppose that a metric space $({\cal M},dist)$ is of doubling dimension $dim$ and let $r>0$ and $\sigma\in(0,1)$.
Then
\begin{itemize}
\item[{\rm({\tt a})}] any ball of radius $r$ in the space $({\cal M},dist)$ can be covered by $(2/\sigma)^{dim}$ balls of radius~$\sigma r$;
\item[{\rm({\tt b})}] a covering of any $r$-radius ball in the space $({\cal M},dist)$ by $(2/\sigma)^{O(dim)}$ balls of radius $\sigma r$ can be computed using $(1/\sigma)^{O(dim)}$ queries to an oracle which returns a covering of arbitrary ball in this space by $2^{O(dim)}$ balls of half the radius.
\end{itemize}
\end{lemma}

\begin{proof}
({\tt a}) Applying the definition of doubling dimension $i$ times, where $i\ge 1$, we conclude that an $r$-radius ball in the space $({\cal M},dist)$ can be covered by $2^{i\cdot dim}$ balls of radius $r/2^i$.
Therefore, by selecting the integer $i$ for which $1/2^i\le\sigma<1/2^{i-1}$, we obtain $2^{i\cdot dim}<(2/\sigma)^{dim}$ covering balls of radius $r/2^i\le\sigma r$.

({\tt b}) Let $i$ be the integer for which $1/2^i\le\sigma<1/2^{i-1}$ and suppose that $j=O(dim)$ is an integer such that the used oracle returns a covering of arbitrary ball in the space $({\cal M},dist)$ by at most $2^j$ balls of half the radius.
Then, by induction, we obtain at most $2^{i\cdot j}<(2/\sigma)^j$ covering balls of radius $r/2^i\le\sigma r$ using $2^{(i-1)\cdot j}<(1/\sigma)^j$ queries to the oracle.
The lemma is proved.
\hfill$\Box$
\end{proof}

\begin{theorem}
For any $\varepsilon\in(0,1]$, every set of $n$ points in a metric space $({\cal M},dist)$ of doubling dimension $dim$ admits a $(1+\varepsilon)$-collection of cardinality $(2/\varepsilon)^{O(dim)}n$.
If an oracle is specified which, in time $2^{O(dim)}$, returns a covering of arbitrary ball in the space $({\cal M},dist)$ by $2^{O(dim)}$ balls of half the radius, then this $(1+\varepsilon)$-collection can be computed by a randomized algorithm in expected time $2^{O(dim)}n\ln n+(2/\varepsilon)^{O(dim)}n$.
\end{theorem}

\begin{proof}
According to Lemma~\ref{lem5}, constructing a $(1+\varepsilon)$-collection for the given set $X$ is reduced to covering each ball $B(x,r_i(k,t))$, $k=1,\dots,s$, $i=1,\dots,I$, $x\in\{a_k,b_k\}$, by balls of radius $\delta d_i(a_k,b_k)(1-2/t)$.
We will use the same values of the parameters $t$ and $I$ as in Sect.~4.3: $I=\big\lceil 1/\log_\varepsilon 0.9\big\rceil$ and $t=\max\big\{10,\,\frac{1+\varepsilon}{\varepsilon}\big\}$.
As shown above, it follows that $\delta\ge 0.9\,\varepsilon$ and $r_i(k,t)\le 2.2\,d_i(a_k,b_k)$.

Note that the $t$-WSPD of the set $X$ constructed by the algorithm from Fact~\ref{f4} is of size $s=t^{O(dim)}n=(2/\varepsilon)^{O(dim)}n$.
On the other hand, Lemma~\ref{lem6} implies that each ball of radius $2.2\,d_i(a_k,b_k)$ in the space $({\cal M},dist)$ admits a covering by $(2/\delta)^{O(dim)}$ balls of radius $\delta d_i(a_k,b_k)(1-2/t)$ and that this covering can be constructed using $(2/\delta)^{O(dim)}$ queries to the specified oracle.
So the total number of covering balls in such coverings of all the balls $B(x,r_i(k,t))$ is $2sI\cdot(2/\delta)^{O(dim)}=(2/\varepsilon)^{O(dim)}n$ and the total time required to construct these coverings is $2sI\cdot(2/\delta)^{O(dim)}\cdot 2^{O(dim)}=(2/\varepsilon)^{O(dim)}n$.
Thus, taking into account the running time of the algorithm from Fact~\ref{f4}, we obtain the resulting time complexity $2^{O(dim)}n\ln n+(2/\varepsilon)^{O(dim)}n$ of constructing a $(1+\varepsilon)$-collection.
The theorem is proved.
\hfill$\Box$
\end{proof}

\section{Conclusion}
We study the concept of a \emph{$(1+\varepsilon)$-approximate centers collection}, an extension of a given set of points which contains a $(1+\varepsilon)$-approximation of each point of space with respect to the distances to all the given points.
We prove that, for any fixed $\varepsilon>0$, eve\-ry set of $n$ points in a metric space of fixed doubling dimension admits a linear-size $(1+\varepsilon)$-collection and that, in many cases, this $(1+\varepsilon)$-collection can be constructed in almost linear time.
It provides a universal discretization of geometric optimization problems reducible to finding points in space (centers) with the minimum value of arbitrary objective function which has a continuity-type dependence on the distances from the centers to given input points.


\providecommand*\hyphen{-}

\end{document}